\documentclass{ieeeaccess}
\usepackage{cite}
\usepackage{amsmath,amssymb,amsfonts}
\usepackage{algorithm}
\usepackage{algorithmic}
\usepackage{graphicx}
\usepackage{textcomp}
\usepackage{times}
\usepackage{tabu}
\usepackage{booktabs, multirow}
\usepackage{soul}
\usepackage{hyperref}
\usepackage{courier}
\usepackage{subfig}
\usepackage[all]{xy}
\usepackage{multirow}
\usepackage{array}
\usepackage{epsfig}
\usepackage{color}
\usepackage{colortbl}
\usepackage[utf8]{inputenc}
\usepackage{kotex}
\usepackage{amsthm}

\newtheorem{theorem}{Theorem}

\usepackage{xspace}
\newcommand{\BfPara}[1]{{\noindent\bf#1.}\xspace}
\usepackage{enumitem}

\begin{document}

\history{Date of publication xxxx 00, 0000, date of current version xxxx 00, 0000.}
\doi{10.1109/ACCESS.2017.DOI}

\title{Joint Mobile Charging and Coverage-Time Extension for Unmanned Aerial Vehicles}
\author{\uppercase{Soohyun Park}\authorrefmark{1}, \uppercase{Won-Yong Shin}\authorrefmark{2} \IEEEmembership{Senior Member, IEEE}, \uppercase{Minseok Choi}\authorrefmark{3}, and \uppercase{Joongheon Kim}\authorrefmark{1} \IEEEmembership{Senior Member, IEEE}}
\address[1]{School of Electrical Engineering, Korea University, Seoul, Korea (e-mails: soohyun828@korea.ac.kr, joongheon@korea.ac.kr)}
\address[2]{School of Mathematics and Computing (Computational Science and Engineering), Yonsei University, Seoul, Korea (e-mail: wy.shin@yonsei.ac.kr)}
\address[3]{Department of Telecommunication Engineering, Jeju National University, Jeju, Korea (e-mail: ejaqmf@jejunu.ac.kr)}
\tfootnote{This research was supported by the National Research Foundation of Korea (NRF) (2021R1A2C3004345, 2020R1G1A1101164), the Institute for Information and Communications Technology Promotion (IITP) funded by MSIT under Grant 2018-0-00170, Virtual Presence in Moving Objects through 5G (2020R1G1A1101164), and the National University Development Project funded by the Ministry of Education (Korea) and NRF (2021).}

\markboth
{S. Park \headeretal: Joint Mobile Charging and Coverage-Time Extension for Unmanned Aerial Vehicles}
{S. Park \headeretal: Joint Mobile Charging and Coverage-Time Extension for Unmanned Aerial Vehicles}

\corresp{Corresponding authors: Won-Yong Shin, Minseok Choi, and Joongheon Kim (e-mails: wy.shin@yonsei.ac.kr, ejaqmf@jejunu.ac.kr, joongheon@korea.ac.kr).}

\begin{abstract}
In modern networks, the use of drones as mobile base stations (MBSs) has been discussed for coverage flexibility. However, the realization of drone-based networks raises several issues. One of critical issues is drones are extremely power-hungry. To overcome this, we need to characterize a new type of drones, so-called charging drones, which can deliver energy to MBS drones. Motivated by the fact that the charging drones also need to be charged, we deploy ground-mounted charging towers for delivering energy to the charging drones. We introduce a new energy-efficiency maximization problem, which is partitioned into two independently separable tasks. More specifically, as our first optimization task, two-stage charging matching is proposed due to the inherent nature of our network model, where the first matching aims to schedule between charging towers and charging drones while the second matching solves the scheduling between charging drones and MBS drones. We analyze how to convert the formulation containing non-convex terms to another one only with convex terms. As our second optimization task, each MBS drone conducts energy-aware time-average transmit power allocation minimization subject to stability via Lyapunov optimization. Our solutions enable the MBS drones to extend their lifetimes; in turn, network coverage-time can be extended.
\end{abstract}


\titlepgskip=-15pt

\maketitle
\section{Introduction}\label{sec:intro}
In modern communication systems, the concept of mobile base stations (MBSs) has been widely and actively discussed in order to establish flexible wireless and cellular networking connections~\cite{cst19saad,twc19saad,wc19saad,wc1,wc2,wc3,infocom19-1}.  
To realize such MBS services, many mobile and network computing systems, i.e., vehicles and drones, have been used~\cite{tvt19shin,tits19-2,spawc18bennis,tmc20,infocom19-2,globecom18}; 1) the vehicles can be used in urban areas to improve the capacity of wireless and cellular systems for serving more users and 2) the drones, also referred to as unmanned aerial vehicles (UAVs), are used for extending the coverage of wireless and cellular services in extreme areas.
Because MBS drone technologies aim at offering robust and flexible connectivity, Internet-of-UAVs (I-UAVs) are widely discussed nowadays. 

However, the MBS drones are extremely energy-limited. In modern embedded drone platforms, the operation time is limited to only a few hours due to battery limitation. Thus, energy-efficient communication plays a crucial role in MBS drones.
The energy-efficient operation is definitely helpful for extending the MBS drone operation times; in turn, it is also useful for extending the network coverage and service time. As discussed in~\cite{ton10,vtc06s}, the concept of coverage-time is defined as the time until one MBS drone totally exhausts its own energy. Therefore, it is obvious that extending the coverage-time is required for network lifetime extension.
To this end, new charging techniques are required. In order to realize seamless MBS operation in drones (i.e., working as MBS drones while being simultaneously charged), wireless charging technologies~\cite{iotj18lee,iotj19lee,iotj2018cho,twc15lee,cm19cho} but also a new concept of mobile drone charging and mobile charging drone concepts (i.e., another type of drones for charging MBS drones such as aerial tankers)~\cite{iotj2018cho} are essentially required.

There have been several recent studies on how to charge MBS drones in the literature. Among them, the use of charging vehicles as proposed in our previous work~\cite{tvt19shin}. Even though it shows reasonable performance improvements, the drones to be charged should move down toward the charging vehicle. Then, the role of MBS drones may not be fulfilled properly.
To overcome the drawback, we utilize a new type of drones, so-called {\em charging drones}, which are designed for providing power sources via wireless charging to MBS drones, which has been widely studied in the literature~\cite{wpt1,ref1}. In addition, due to the fact that the charging drones are also battery-operated, they also need to be energy-efficient with a demand of being charged. This motivates us to characterize a charging infrastructure that is composed of charging towers in this paper. Since the charging towers are ground-mounted, they are capable of acquiring power sources without strict limitations. Note that, due to the fact that the towers have charging plates, the charging drones can be served via wireless charging.
Such a mobile drone charging system can be widely used as the next-generation UAV model for surveillance and seamless Internet connection in many extreme wilds where human access may be unavailable, \textit{e.g.} mountainous areas or vast coastlines.

In this paper, we propose a new energy-efficiency maximization problem, which is partitioned into two independently separable optimization tasks, under our network model. Toward this end, as our first optimization task, we design two-stage charging matching/scheduling, i.e., (i) charging matching between charging towers and charging drones and (ii) charging matching between charging drones and MBS drones. For the second case, in addition to the matching decision, the amounts of allocating power from charging drones to matched MBS drones should be optimally computed due to the limited power in charging drones. To solve this issue, we first prove that the problem is non-convex, thereby converting the original non-convex formulation into a convex setting that can guarantee optimal solutions.
After conducting this two-stage charging matching/scheduling, as our second optimization task we design energy-aware data transmission in each MBS drone due to the fact that the drones are power-hungry. Inspired by the Lyapunov optimization~\cite{book2010sno}, a time-average energy consumption minimization framework is proposed by controlling the transmit power in each MBS drone under system stability. We note that the beauty of this framework is the realization of distributed operations, i.e., sharing information among MBS drones is not required. 

\BfPara{Related Work}
As we mentioned above, drones widely used for many applications, \textit{e.g.} video provisioning~\cite{icon19kwon}, mobile edge computing~\cite{iwcmc19}, and aerial information sensing~\cite{infocom19-3}, can be a key solution for MBS services in cellular networks.
In order to use the drones for cellular networks, theoretical performance analysis under consideration of practical antenna configurations~\cite{rev1} and mobility patterns~\cite{rev2} is essentially required.

There are have been a great deal of studies on drone trajectory optimization for various objectives.
Most researches aim at carrying out transmission rate maximization, transmit power control for energy-efficiency, and latency minimization.
The research in~\cite{spawc18} aims at the maximization of cumulative information rate (i.e., transmission sum rate) by optimally controlling the direction of drones using reinforcement learning (RL). 
If there are more than two users, then the drone will move through the optimal choice via various Q-learning algorithms.
However, since researchers in~\cite{spawc18} do not consider drone's energy, they get the optimal trajectories for two users considering communication obstacles.
In~\cite{twc19}, researchers formulated a problem to find the optimal drone trajectory in multi-drone environments using a dynamic non-cooperative game theory; however, a charging technique was also excluded.
There was a study that attempts to conduct joint power and trajectory optimization, where drones engage in cooperative communications on amplify-and-forward mode~\cite{wc2018}. 
As a result, the drone finds an optimal location, \textit{i.e.}, the distance between users and the BS, that minimizes outage probabilities. 
In addition, power control can be performed to consume less communication power/energy between the drones which are close to each other~\cite{p1,p2,p3}. 
Nevertheless, the aforementioned studies do not take into account the concept of mobile drone charging, which can be thought of as a practically important scenario.

\BfPara{Contributions} The main contributions of this research are four-fold and are summarized as follows.
\begin{itemize}
    \item It is the first attempt to characterize a new mobile UAV charging system aimed at charging both MBS drones and charging drones, which is a more feasible scenario in practice. 
    \item In the model, we formulate two-stage optimization problems for charging the two-types of drones. We analyze that the optimization problem for scheduling between charging drones and MBS drones is non-convex. We then present a method of transforming the non-convex formulation into  the convex setting, thus resulting in achieving optimal solutions.
    \item In addition, we perform distributed time-average transmit power allocation at each MBS drone subject to queue stability by virtue of Lyapunov optimization theory. Our method is beneficial in the sense of guaranteed queue stability and fully distributed operation at MBS drones.
    \item Through data-intensive simulations, it is demonstrated that our matching method remarkably outperforms several baseline schemes in terms of the average residual energy and the fairness among served drones. The coverage-time and queue stability are also evaluated to validate the effectiveness of our time-average transmit power allocation method at each MBS drone.
\end{itemize}

\BfPara{Organization} The rest of this paper is organized as follows. Section~\ref{sec:2} presents the reference network model and related work. 
The proposed two-stage mobile charging and traffic-aware coverage-time extension method is explained in Section~\ref{sec:3}.
Section~\ref{sec:4} evaluates the performance of the proposed method.
Lastly, Section~\ref{sec:4} concludes this paper.

\begin{figure}[t]
    \centering
        \includegraphics[width =0.99\linewidth]{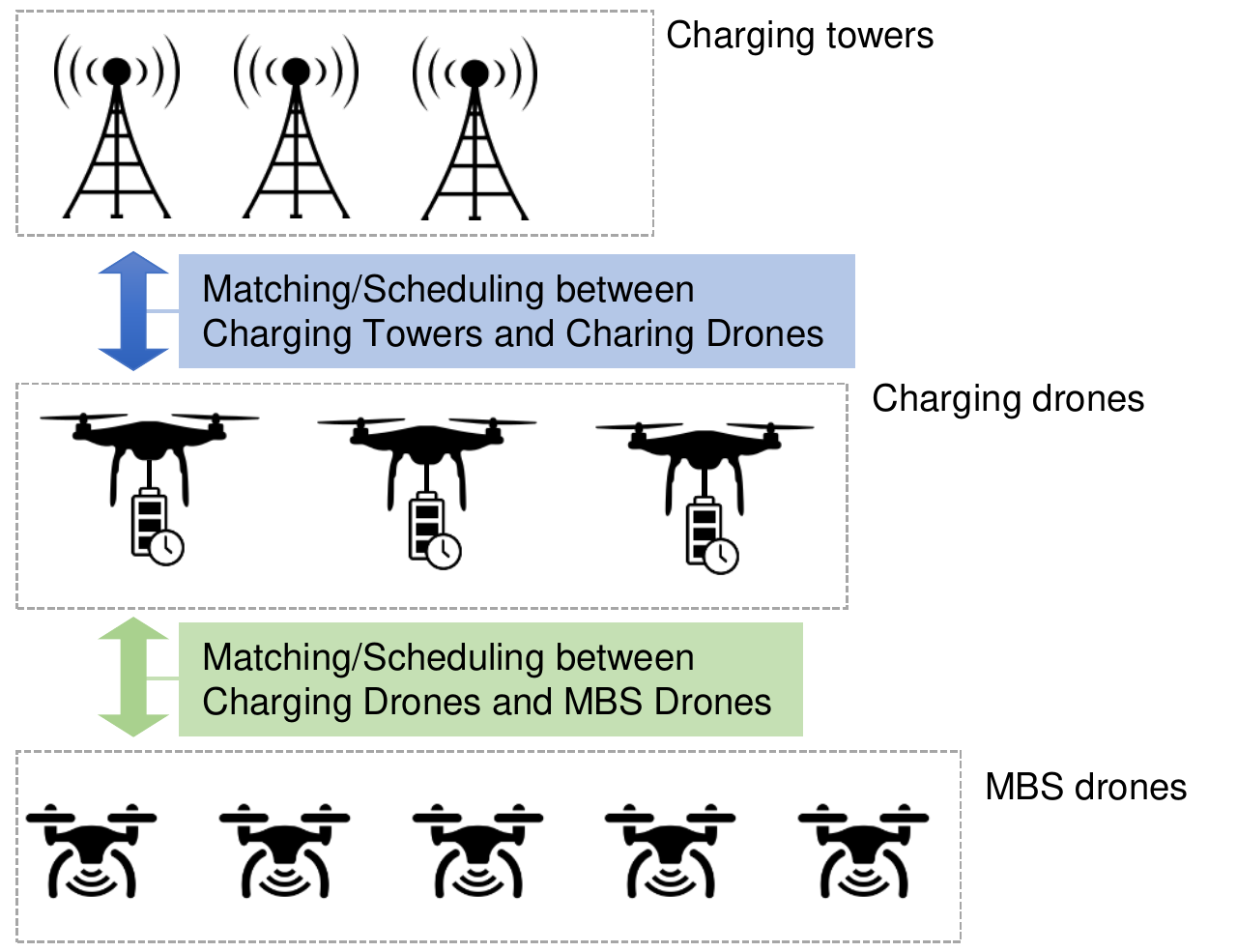}
    \caption{Reference network model where three types of elements are deployed. Here, two-stage matching/scheduling is required due to the inherent nature of our network architecture.}
    \label{fig:fig_netmodel}
\end{figure}

\section{System Model}\label{sec:2}
As a potential UAV configuration, we assume our network model consisting of three types of elements, i.e., charging towers, charging drones, and MBS drones, as shown in Fig.~\ref{fig:fig_netmodel}. 

First of all, the charging tower is a static and ground-mounted infrastructure which can provide power sources for charging drones. The main objective and purpose of the charging towers is for delivering power/energy sources to MBS drones via charging drones for extending cellular network coverage-time. 
Because the towers are ground-mounted, they can get power sources without limitations. Note that the charging towers are equipped with several charging plates, thus they can charge multiple charging drones simultaneously.

Next, the charging drones are flying moving objects which can be charged by charging towers and then provide their charged powers to their matched/scheduled MBing drones. 
The main objective is for providing powers to MBS drones, as mentioned. 
Because they are mobile and battery-operated, energy-efficient operations are essential. 

Lastly, the MBS drones work as mobile cellular base stations, in specific fixed-position cellular network coverage areas, which can be charged by charging drones.
Thus, we design the MBS drones to receive energy supply from charging towers through the charging drones, where the MBS drones and charging towers are fixed in their own positions.
Because they are also battery-operated, energy-efficient operations are obviously required.

In this paper, the sets of charging towers, charging drones, and MBS drones are denoted as $\mathcal{T}$, $\mathcal{C}$, and $\mathcal{M}$, respectively.

\section{Joint Mobile Charging and Coverage-Time Extension}\label{sec:3}

This section consists of algorithm design concepts and rationale (refer to Sec.~\ref{sec:3-1}) and the details (refer to Sec.~\ref{sec:3-2}).
\subsection{Algorithm Design Rationale}\label{sec:3-1}
The proposed method in this paper aims at maximizing energy-efficiency of our UAV charging system, which consists of charging towers, charging drones, and MBS drones, subject to system stability. To this end, we formulate our problem:

\begin{multline}
    \max: \underbrace{\Psi_{(j,k)}\left(E_{j}^{c},e_{j}^{c}[t]\right) + \Psi_{(i,j)}\left(E_{j}^{c},E_{i}^{m},e_{j}^{c}[t],e_{i}^{m}[t]\right)}_{\text{mobile charging matching}}
    \\
    - \underbrace{\Phi(\alpha[t])}_{\text{dynamic TX power allocation}}
    \label{eq:total}
\end{multline}

where $\forall i\in\mathcal{M}$, $\forall j\in\mathcal{C}$, and $\forall k \in \mathcal{T}$ as $\mathcal{M}$, $\mathcal{C}$ and $\mathcal{T}$ are the sets of MBS drones, charging drones, and charging towers, respectively;
$\Psi_{(j,k)}\left(E_{j}^{c},e_{j}^{c}[t]\right)$ is the matching function between charging tower $k\in\mathcal{T}$ and charging drone $j\in\mathcal{C}$ so that the energy delivered from charging tower $k$ to charging drone is maximized when $E_{j}^{c}$ and $e_{j}^{c}[t]$ are defined as the energy capacity (i.e., the amount of energy when the drone is fully charged) at charging drone $j$ and the residual energy at charging drone $j$, respectively; and $\Psi_{(i,j)}\left(E_{j}^{c},E_{i}^{m},e_{j}^{c}[t],e_{i}^{m}[t]\right)$ is the matching function between charging drone $j\in\mathcal{C}$ and MBS drone $i\in\mathcal{M}$ so that the energy delivered from charging drone $j$ to MBS drone $i$ is maximized when $E_{i}^{m}$
and $e_{i}^{m}[t]$ are defined as the energy capacity at MBS drone $i$
and the residual energy at MBS drone $i$, respectively. Thus, after the optimal two-stage mobile charging matching/scheduling, the sum of the first two terms in \eqref{eq:total},  $\Psi_{(j,k)}\left(E_{j}^{c},e_{j}^{c}[t]\right) + \Psi_{(i,j)}\left(E_{j}^{c},E_{i}^{m},e_{j}^{c}[t],e_{i}^{m}[t]\right)$, represents the overall maximum energy delivery from energy sources to individual MBS drones. In addition, the third term in \eqref{eq:total}, $\Phi(\alpha[t])$, is the time-average energy consumption function with respect to transmit power allocation action $\alpha[t]$ at time $t$. Then, we aim at minimizing $\Phi(\alpha[t])$ at individual MBS drones subject to system/queue stability by controlling $\alpha[t]$ at each time. We shall specify the three functions ($\Psi_{(j,k)}, \Psi_{(i,j)}, and, \Phi(\alpha[t])$) in Section III-B where each mathematical program is provided along with its algorithm details.

It is worth noting that the mobile charging matching is carried out in terms of the matching indices between $\forall k\in\mathcal{T}$ and $\forall j\in\mathcal{C}$ and between $\forall j\in\mathcal{C}$ and $\forall{i}\in\mathcal{M}$, whereas the transmit power allocation is performed in terms of the allocation decision $\alpha[t]$. This implies that the above two tasks, corresponding to $\Psi_{(j,k)}\left(E_{j}^{c},e_{j}^{c}[t]\right) + \Psi_{(i,j)}\left(E_{j}^{c},E_{i}^{m},e_{j}^{c}[t],e_{i}^{m}[t]\right)$ and $\Phi(\alpha[t])$, do not share any control parameters and control actions are independent of each other. Thus, those two optimization tasks can be linearly and independently separable.

The proposed method in this paper is defined in linearly separable two algorithms, i.e., (i) two-stage mobile charging matching and (ii) distributed time-average energy-efficient transmission optimization (i.e., time-average transmit power minimization) subject to stability.

\begin{itemize}
    \item For the first, we design our charging matching method in two-stage, i.e., matching between charging towers and charging drones and matching between charging drones and MBS drones. The reason why joint optimization is not chosen is that the different type of drones can be operated by different service providers, thus, separable operations are needed for scalability. More details are in Sec.~\ref{sec:algo1}.
    \item For the second, in order to achieve more battery-saving, energy-efficient operations are essential. Thus, a dynamic and distributed energy-aware transmit power allocation algorithm is designed. The algorithm dynamically works based on queue-backlog status, and this concept is formulated via Lyapnov optimization theory.
More details are in Sec.~\ref{sec:algo2}.
\end{itemize}

The proposed method assumes that the joint mobile charging and coverage-time extension computation conducts at a centralized controller which is connected to charging towers in order to guarantee the stabilized power supply (refer to \cite{wc1,tvt19shin,tmc20,infocom19-2,icon19kwon,iwcmc19,spawc18,wc2018} and references therein for the detailed description).

\begin{figure}[t]
    \centering
        \includegraphics[width =0.99\linewidth]{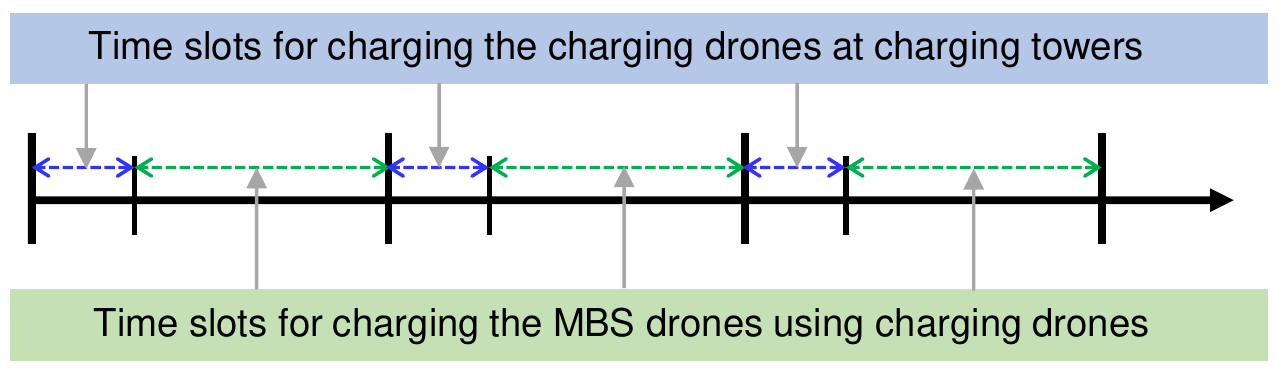}
    \caption{Unit time organization: In each unit time $u$, two-stage operation conducts (i.e., $u=u^{c}+u^{m}$), i.e., (i) $u^{c}$ for charging the charging drones $\forall j\in\mathcal{C}$ and at charging towers $\forall k\in\mathcal{T}$ and (ii) $u^{m}$ for charging MBS drones $\forall i\in\mathcal{M}$ using scheduled/matched charging drones $\forall j\in\mathcal{C}$.}
    \label{fig:fig_timescale}
\end{figure}

\subsection{Algorithm Details}\label{sec:3-2}
As discussed, the proposed method consists of two algorithms, i.e., (i) two-stage mobile charging matching/scheduling (refer to Sec.~\ref{sec:algo1}) and (ii) energy-aware transmit power allocation subject to queue stability at MBS drones (refer to Sec.~\ref{sec:algo2}). Lastly, the computational complexity of the proposed method is presented (refer to Sec.~\ref{sec:algo-complexity}).

\subsubsection{Two-stage mobile charging matching/scheduling}\label{sec:algo1}
The proposed charging matching consists of two-stages as follows:
\begin{itemize}
    \item Matching between charging towers and charging drones
    \item Matching between charging drones and MBS drones
\end{itemize}

Each matching stage conducts within a single unit time $u$, as illustrated in Fig.~\ref{fig:fig_timescale}. 
Each stage is explained as follows, for more detailed optimization solving procedures.

\BfPara{Matching/scheduling between charging towers and charging drones}
Given $E_j^c$ and $e_j^c[t]$, the matching function $\Psi_{(j,k)}(E_j^c,e_j^c[t])$ is materialized below. Specifically, mathematical program for matching between charging drones and charging towers can be formulated as follows:

\begin{eqnarray}
    \max: & & \sum_{\forall j\in\mathcal{C}} \sum_{\forall k\in\mathcal{T}} 
    \left(E_{j}^{c}-e_{j}^{c}[t]\right)\cdot x_{(k,j)}^{t}[t]
    \label{eq:s1obj} \\
    \text{s.t.} & & 
    \sum_{\forall j\in\mathcal{C}} x_{(k,j)}^{t}[t] \leq n_{k}^{t}, \forall k\in \mathcal{T} \label{eq:s1c1} \\
     & & 
    \sum_{\forall k\in\mathcal{T}} x_{(k,j)}^{t}[t] \leq 1, \forall j\in \mathcal{C} \label{eq:s1c2}
\end{eqnarray}
where 
$x_{(k,j)}^{t}[t]$ is the matching/scheduling index where it is $1$ when the charging  tower $k$ where $\forall k \in\mathcal{T}$ is selected to provide powers to charging drone $j$ where $\forall j \in \mathcal{C}$ (whereas it is $0$ when $k$ and $j$ are not matched/scheduled) via wireless charging technologies, and 
$n_{k}^{t}[t]$ is the number of charging plates at charging tower $k$ where $\forall k \in\mathcal{T}$, respectively.
Note that all parameters in (\ref{eq:s1obj})--(\ref{eq:s1c2}) except $x_{(k,j)}^{t}[t]$ where $\forall j\in \mathcal{C}$, $\forall k\in\mathcal{T}$ are constants, thus, this mathematical program is for computing optimal $x_{(k,j)}^{t}[t]$ where $\forall j\in \mathcal{C}$, $\forall k\in\mathcal{T}$, which can maximize (\ref{eq:s1obj}). 

In (\ref{eq:s1obj}), $\left(E_{j}^{c}-e_{j}^{c}[t]\right)$ means the difference between full energy capacity and current residual energy at charging drone $j$. Thus, the amount means how much energy can be charged for the charging drone $j$ (which is defined as \textit{charging capacity} at charging drone $j$). As presented in (\ref{eq:s1obj}), the matching/scheduling between charging towers and charging drones is for maximizing the summation of the charging capacity of matched/scheduled charging drones. 

In this mathematical program, we have to state that each charging tower can charge $n_{k}^{t}$ number of charging drones where $n_{k}^{t}$ is the number of charging plates at charging tower $k$ where $\forall k \in \mathcal{T}$; and it is formulated as a constraint in (\ref{eq:s1c1}). Moreover, each charging drone can be served by a single charging tower, as formulated in (\ref{eq:s1c2}).

Once this matching/scheduling decision is made, each charging tower starts the charging operation for its matched/scheduled charging drone. The amount of charging energy via wireless charging is determined as follows:

\begin{equation}
    \mathcal{E}_{k}^{t}\cdot \eta_{k}^{t}\cdot\eta_{j}^{c}\cdot \left(u^{t}-\frac{d_{(k,j)}[t]}{s_{j}}\right)
    \label{eq:charging1}
\end{equation}
where $\mathcal{E}_{k}^{t}$ is the amount of charging energy at charging tower $k$ where $\forall k\in\mathcal{T}$, 
$\eta_{k}^{t}$ is wireless charging efficiency at charging tower $k$ where $\forall k\in\mathcal{T}$, 
$\eta_{j}^{c}$ is the wireless charging efficiency of charging drone $\forall j \in\mathcal{C}$, 
$u^{t}$ is the operation time for this matching between charging towers and charging drones, 
$d_{(k,j)}[t]$, is the distance between charging tower $k$ where $\forall k\in\mathcal{T}$ and charging drone $j$ where $\forall j\in\mathcal{C}$,
and 
$s_{j}$ is the speed of charging drone $j$ where $\forall j \in\mathcal{C}$. 

Then, the final energy status in each matched/scheduled charging drone $j$ where $\forall j\in\mathcal{C}$ is as follows:

\begin{equation}
    \min\left\{e_{j}^{c}[t] + \underbrace{\mathcal{E}_{k}^{t}\cdot \eta_{k}^{t}\cdot \eta_{j}^{c}\cdot\left(u^{t}-\frac{d_{(k,j)}[t]}{s_{j}}\right)}_{\text{Amount of charging via wireless charging}}, E_{j}^{c}\right\}
    \label{eq:charging2}
\end{equation}
due to the fact that the summation of current residual energy and charged energy cannot exceed the full energy capacity at charging drone $j$ where $\forall j\in\mathcal{C}$. 

\BfPara{Matching/scheduling between charging drones and MBS drones}
Given $E_j^c$, $E_i^m$, $e_j^c[t]$, and $e_i^m[t]$, another matching function $\Psi_{(i,j)}(E_j^c,E_i^m,e_j^c[t],e_i^m[t])$ is materialized below. Specifically, mathematical program for matching between charging drones and MBS drones can be formulated as follows:

\begin{eqnarray}
    \max: & & \sum_{\forall j\in\mathcal{C}} \sum_{\forall i\in\mathcal{M}} 
    v_{(i,j)}[t] \cdot x_{(i,j)}^{m}[t]
    \label{eq:s2obj} \\
    \text{s.t.} & & 
    \sum_{\forall j\in\mathcal{C}} e_{(i,j)}[t] \cdot \eta_{j}^{c}\cdot
    \eta_{i}^{m}\cdot x_{(i,j)}^{m}[t] \leq \nonumber \\
    & & \quad\quad\quad\quad\quad\quad\quad\quad E_{i}^{m}-e_{i}^{m}[t], \forall i\in \mathcal{M} \label{eq:s2c1} \\
     & & 
    \sum_{\forall j\in\mathcal{C}} x_{(i,j)}^{m}[t] \leq n_{i}^{m}, \forall i\in \mathcal{M} \label{eq:s2c2} \\
     & & 
    \sum_{\forall i\in\mathcal{M}} x_{(i,j)}^{m}[t] \leq 1, \forall j\in \mathcal{C} \label{eq:s2c3} \\
     & & 
     \sum_{\forall i\in\mathcal{M}} e_{(i,j)}[t] \leq e_{j}^{c}[t], \forall j\in \mathcal{C} \label{eq:s2c4} \\
     & & 
     e_{(i,j)}[t] \geq 0, \forall i\in \mathcal{M}, \forall j\in \mathcal{C} \label{eq:s2c5}
\end{eqnarray}
where $e_{(i,j)}[t]$ is the energy to charge for MBS drone $i$ where $\forall i \in\mathcal{M}$ via charging drone $j$ where $\forall j \in\mathcal{C}$ which we have to calculate, i.e., decision variables, and 
$x_{(i,j)}^{m}[t] \in \mathcal{F}_{(i,j)}[t]$ where $\forall i\in \mathcal{M}$, $\forall j\in \mathcal{C}$, respectively.
Here, $\mathcal{F}_{(i,j)}[t]$ is the set of feasible matching between charging drones $j$ where $\forall j\in\mathcal{C}$ and MBS drones $i$ where $\forall i\in\mathcal{M}$ at time $t$. 

In (\ref{eq:s2obj}), $v_{(i,j)}[t]$ is defined as the value while a charging drone $j$ where $\forall j \in\mathcal{C}$, is scheduled for charging a mobile BS drone $i$ where $\forall i \in\mathcal{M}$, via wireless charging technologies. The $v_{(i,j)}[t]$ can be formulated as follows:
\begin{equation}
    v_{(i,j)}[t] =
    \left(u^{m}-\frac{d_{(i,j)}[t]}{s_{j}}\right)\cdot
    \eta_{j}^{c}\cdot
    \eta_{i}^{m}\cdot
    \frac{e_{(i,j)}^{c,\max}}{e_{i}^{m}[t]}\cdot 
    \left(e_{j}^{c}[t]-e_{(i,j)}[t]\right)
    \label{eq:value}
\end{equation}
where $e_{(i,j)}^{c,\max}\triangleq\max\left\{e_{j}^{c}[t]-e_{(i,j)}^{c}[t],0\right\}$
and
$u^{m}$ is operation time for this mobile charging matching,
$d_{(i,j)}[t]$ is a distance between charging drone $\forall j \in\mathcal{C}$ and MBS drone $\forall i \in\mathcal{M}$,
$\eta_{i}^{m}$ is the wireless charging efficiency of mobile BS drone $i$ where $\forall i \in\mathcal{M}$, 
and
$e_{j}^{c}[t]$ is the residual energy at charging drone $j$ where $\forall j \in\mathcal{C}$,
respectively.
$e_{(i,j)}^{c}[t]$ is defined as the energy consumption while a charging drone $\forall j \in\mathcal{C}$ approaches an MBS drone $\forall i \in\mathcal{M}$. It can be formulated as follows:
\begin{equation}
    e_{(i,j)}^{c}[t] = 
    \int_{0}^{\Delta t_{(i,j)}}
    f_{e}(t)dt
\end{equation}
where $\Delta t_{(i,j)}\triangleq \frac{d_{(i,j)}[t]}{s_{j}}$ and $f_{e}(t)$ stands for the energy expenditure in time $t$.
Here, $\max\left\{e_{j}^{c}[t]-e_{(i,j)}^{c}[t],0\right\}$ means that if current charging drone $j$ where $\forall j\in\mathcal{C}$ does not have enough residual energy to move to MBS drone $i$ where $\forall i\in\mathcal{M}$, the value becomes $0$, in turn, $v_{(i,j)}[t]=0$. Finally, the charging drone $j$ where $\forall j\in\mathcal{C}$ cannot be matched with the MBS drone $i$ where $\forall i\in\mathcal{M}$ because $v_{(i,j)}[t]=0$ (semantically, the charging drone does not have enough energy to move to the MBS drone). 
Note that all variables in (\ref{eq:value}), except $e_{(i,j)}[t]$, are constants.

The summation of charged energy by each matched/scheduled charging drones for one MBS drone $i$ where $\forall i\in\mathcal{M}$ cannot exceed the charging capacity, i.e., the difference between full energy capacity and current residual energy at time $i$ (i.e., $E_{i}^{m}-e_{i}^{m}[t]$ where $\forall i\in\mathcal{M}$). This is formulated in \eqref{eq:s2c1}, and then it can be seen that the charged energy from charging drone $j$ to MBS drone $i$ can be formulated as $e_{(i,j)}[t] \cdot \eta_{j}^{c}\cdot \eta_{i}^{m}$.

In terms of matching, each MBS drone $i$ can be matched/scheduled with $n_{i}^{m}$ charging drones where $n_{i}^{m}$ stands for the charging plates at the MBS drone $i$ where $\forall i \in\mathcal{M}$, as stated in \eqref{eq:s2c2}. 
Similarly, each charging drone $i$ can serve only one MBS drone, as stated in \eqref{eq:s2c3}. 

The energy charging amount at charging drone $j$ where $\forall j\in\mathcal{C}$ should between $0$ and current residual energy status at time $t$, $e_{j}^{c}[t]$ where $\forall j\in\mathcal{C}$, as formulated in \eqref{eq:s2c4} and \eqref{eq:s2c5}.

\begin{theorem}
The mathematical program in (\ref{eq:s2obj})--(\ref{eq:s2c5}) is non-convex optimization. 
\label{thm:thm1}
\end{theorem}

\begin{proof}
Here, we have to prove that (\ref{eq:s2obj}) and (\ref{eq:s2c1}) are not convex in this mathematical program (\ref{eq:s2obj})--(\ref{eq:s2c5}).
Note that the proof considers the simplest case at first, i.e., $|\mathcal{C}|=|\mathcal{M}|=1$.
In this case, the (\ref{eq:s2obj}) becomes $v_{(i,j)}[t]\cdot x_{(i,j)}^{m}[t]$ where $v_{(i,j)}[t]$ is defined in (\ref{eq:value}); and let this equation be denoted by $f_{1}$. To show that this given equation is non-convex, the second-order Hessian of this given real function should be non-positive definite~\cite{convex,tbc2013joongheon,icc2013joongheon,usc2014joongheon}. The Hessian $\bigtriangledown^{2}f_{1}$ is as follows where the two variables in (\ref{eq:s2obj}) are $e_{(i,j)}[t]$ and $x_{(i,j)}^{m}[t]$:
$\begin{bmatrix}
0  & -1 \\
-1 & 0 
\end{bmatrix}
$ and then the corresponding two eigenvalues are $\pm 1$. 
These values are not all non-negative, which shows that the Hessian is not positive definite, thus, finally it proves that the optimization function is non-convex.

For (\ref{eq:s2c1}), in a similar way, its second-order Hessian matrix is as
$
\begin{bmatrix}
0  & \eta_{j}^{c}\cdot
    \eta_{i}^{m} \\
\eta_{j}^{c}\cdot
    \eta_{i}^{m} & 0 
\end{bmatrix}
$, and then the corresponding two eigenvalues are $\pm \left(\eta_{j}^{c}\cdot
    \eta_{i}^{m}\right)$, and thus the values are not all non-negative, i.e., it also proves that the optimization function is non-convex.

Finally, in our mathematical program (\ref{eq:s2obj})--(\ref{eq:s2c5}), two terms, i.e.,  (\ref{eq:s2obj}) and (\ref{eq:s2c1}), are not non-convex. Thus, our given mathematical program (\ref{eq:s2obj})--(\ref{eq:s2c5}) is not convex.
\end{proof}

According to the fact that optimal solutions cannot be obtained in non-convex programming, it is required to convert this non-convex programming to convex programming, if possible.

\begin{theorem}
\label{thm:thm2}
For the given non-convex optimization program (\ref{eq:s2obj})--(\ref{eq:s2c5}), introducing
\begin{multline}
        \sum_{\forall j\in\mathcal{C}} e_{(i,j)}[t] \cdot \eta_{j}^{c}\cdot
    \eta_{i}^{m} \leq \\ \left(E_{i}^{m}-e_{i}^{m}[t]\right)\cdot x_{(i,j)}^{m}[t], \forall i\in \mathcal{M}
    \label{eq:thm}
\end{multline}
instead of (\ref{eq:s2c1}) makes the program convex.
\end{theorem}

\begin{proof}
For the non-convex program (\ref{eq:s2obj})--(\ref{eq:s2c5}), 
$x_{(i,j)}^{m}[t]=0$ means the matching between charging drone $j$ and MBS drone $i$ does not happen. Thus, the corresponding charging is not occurred, i.e., $e_{(i,j)}=0$ and (\ref{eq:thm}) leads to the same result when $x_{(i,j)}^{m}[t]=0$, i.e.,
\begin{equation}
        \sum_{\forall j\in\mathcal{C}} e_{(i,j)}[t] \cdot \eta_{j}^{c}\cdot
    \eta_{i}^{m} \leq \left(E_{i}^{m}-e_{i}^{m}[t]\right)\cdot \underbrace{0}_{x_{(i,j)}^{m}[t]=0} =  0,  
\end{equation}
$\forall i\in \mathcal{M}$ and thus,
\begin{equation}
        \sum_{\forall j\in\mathcal{C}} e_{(i,j)}[t] \leq 0, \forall i\in \mathcal{M} 
\end{equation}
because  $\eta_{j}^{c}$ and $\eta_{i}^{m}$ are positive, then, $e_{(i,j)}[t]=0$ because $e_{(i,j)}[t]$ is also non-negative. 

On the other hand, if $x_{(i,j)}^{m}[t]=1$, this (\ref{eq:thm}) is equivalent to (\ref{eq:s2c1}).
Therefore, in turn, (\ref{eq:s2obj}) is updated as follows: 
\begin{equation}
    \max: \sum_{\forall j\in\mathcal{C}} \sum_{\forall i\in\mathcal{M}} 
    v_{(i,j)}[t].
\end{equation}

Then, it is obvious that there are no non-convex terms in this proposed program (\ref{eq:s2obj})--(\ref{eq:s2c5}).
\end{proof}

Eventually, the final form of mobile charging scheduling for multi-drone mobile base station life time extension can be expressed as follows:

\begin{eqnarray}
    \max: & & \sum_{\forall j\in\mathcal{C}} \sum_{\forall i\in\mathcal{M}} 
    v_{(i,j)}[t]
    \label{eq:s3obj} \\
    \text{s.t.} & & 
    \sum_{\forall j\in\mathcal{C}} e_{(i,j)}[t] \cdot \eta_{j}^{c}\cdot
    \eta_{i}^{m} \leq \nonumber \\
    & & \quad\quad\quad \left(E_{i}^{m}-e_{i}^{m}[t]\right)\cdot x_{(i,j)}^{m}[t], \forall i\in \mathcal{M} \label{eq:s3c1} \\
     & & 
    \sum_{\forall j\in\mathcal{C}} x_{(i,j)}^{m}[t] \leq n_{i}^{m}, \forall i\in \mathcal{M} \label{eq:s3c2} \\
     & & 
    \sum_{\forall i\in\mathcal{M}} x_{(i,j)}^{m}[t] \leq 1, \forall j\in \mathcal{C} \label{eq:s3c3} \\
     & & 
     \sum_{\forall i\in\mathcal{M}} e_{(i,j)}[t] \leq e_{j}^{c}[t], \forall j\in \mathcal{C} \label{eq:s3c4} \\
     & & 
     e_{(i,j)}[t] \geq 0, \forall i\in \mathcal{M}, \forall j\in \mathcal{C} \label{eq:s3c5}
\end{eqnarray}
where $v_{(i,j)}[t]$ is defined in (\ref{eq:value}) and our control/decision variables are $e_{(i,j)}[t]$ and $x_{(i,j)}^{m}[t]$, $\forall i\in\mathcal{M}$, $\forall j\in\mathcal{C}$. 
Note that the equations, i.e., (\ref{eq:s3obj})--(\ref{eq:s3c5}), in this mathematical program are all linear combinations, by Theorem~\ref{thm:thm1} and Theorem~\ref{thm:thm2}.

\begin{figure}[t]
    \centering
        \includegraphics[width =0.99\linewidth]{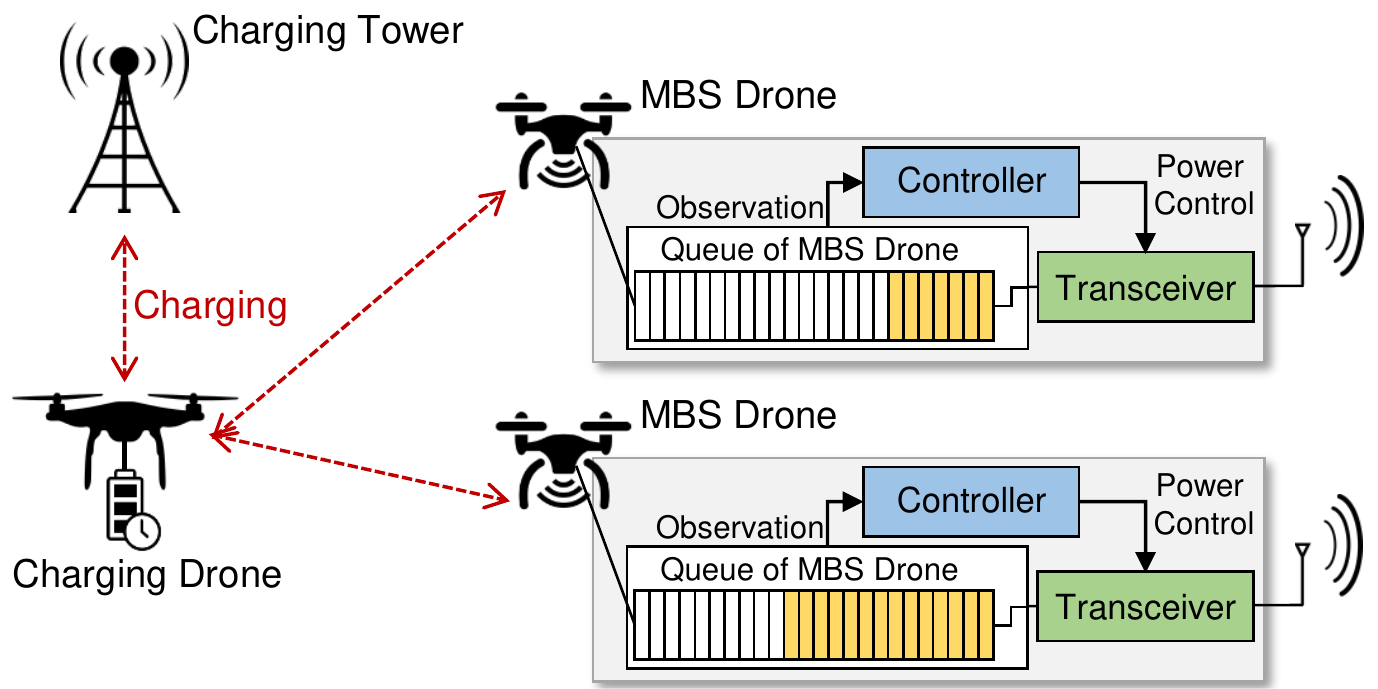}
    \caption{Illustration of MBS drone communication architecture.}
    \label{fig:fig_lyapunov}
\end{figure}

\subsubsection{Delay-aware time-average optimal control via Lyapunov optimization}\label{sec:algo2}

In our system, MBS drones essentially require energy-efficient operation because they are power hungry. In each MBS drone $i$ where $\forall i\in\mathcal{M}$, allocating small transmit power is helpful in terms of energy-efficient operations. However, it transmits relatively small number of packets over the air which introduces delays with the queue of the MBS drone $i$ where $\forall i\in\mathcal{M}$. On the other hand, large transmit power allocation introduces more packet transmission whereas it consumes more energy sources.
Eventually, we can observe the tradeoff between our objective (i.e., energy-efficiency) and stability. This paper designs a dynamic power allocation method for energy-efficient operations while guaranteeing queue stability. Note that this kind of delay-aware time-average optimization is definitely needed for computing system design and implementation~\cite{tvt2019minseok,ton2016joongheon,mm2017jonghoe,tmc2019jonghoe}.
Our considering network architecture is as shown in Fig.~\ref{fig:fig_lyapunov}.
As explained, each MBS drone observes its own queue and then its own controller makes the decision for transmit power control/allocation to its own transceiver aiming at time-average energy-efficiency maximization subject to queue stability. 

We formulate the mathematical program for minimizing the time-average energy consumption for processing packets from the queue, i.e., $E(\alpha[t])$, where the power allocation decision is $\alpha[t]$, can be presented as follows:
\begin{eqnarray}
\min: & & \lim_{t\rightarrow\infty}\sum_{\tau=0}^{t-1} E(\alpha[\tau])
\label{eq:opt} \\
\text{s.t.} & & \lim_{t\rightarrow\infty}\frac{1}{t}\sum_{\tau=0}^{t-1} Q[\tau]<\infty \text{ (stability constraint)}
\end{eqnarray}

In \eqref{eq:opt}, $E(\alpha[t])$ stands for the energy consumption for queue departure process $b(\alpha[t])$ when the given power allocation decision is $\alpha[t]$. 
As mentioned earlier, the power allocation decision generates a tradeoff between the minimization of energy consumption and stability of the queuing system which is related to the average queuing delay. 

Respect to this tradeoff, the Lyapunov optimization theory-based drift-plus-penalty (DPP) algorithm~\cite{book2010sno,tvt2019minseok, ton2016joongheon, mm2017jonghoe,tmc2019jonghoe,isj2020soyi} can be used for optimizing the time-average utility function (i.e., energy consumption) subject to queue stability.
Define the Lyapunov function 
$L(Q[t]) = \frac{1}{2}(Q[t])^2$,
and let $\Delta(.)$ be a conditional quadratic Lyapunov function that can be formulated as 
\begin{equation}
\mathbb{E}[L(Q[t+1])-L(Q[t])| Q[t]], 
\end{equation}
called as the drift on $t$. 
After the MBS drone $i$ where $\forall i \in\mathcal{M}$ observes the current queue length $Q(t)$, how much power sources are required in each time slot.
According to \cite{book2010sno}, this dynamic policy is designed to achieve queue stability by minimizing an upper bound on drift-plus-penalty which is given by
\begin{equation}
    \Delta(Q[t]) + V \mathbb{E} \Big[ E(\alpha[t]) \Big],
\end{equation}
where $V$ is an importance weight for power efficiency. 

The upper bound on the drift of the Lyapunov function on $t$ can be derived as follows:
\begin{eqnarray}
    L(Q[t+1]) - L(Q[t]) &=& \frac{1}{2}\Big( Q([t+1]^2 - Q[t]^2 \Big) \\
    &\leq& \frac{1}{2} \Big( a[t]^2 + b(\alpha[t])^2 \Big) + \\
    & & Q[t] (a[t] - b(\alpha[t])).
\end{eqnarray}

Therefore, the upper bound on the conditional Lyapunov drift is obtained as follows:
\begin{align}
    \Delta(Q(t)) &= \mathbb{E}[L(Q[t+1]) - L(Q[t]) | Q[t]] \nonumber \\
    &\leq C + \mathbb{E}\Big[ Q[t](a[t] - b(\alpha[t]) \Big| Q[t] \Big],
\end{align}
where $C$ is a constant which can be given by
\begin{equation}
    \frac{1}{2}\mathbb{E}\Big[ a[t]^2 + b(\alpha[t])^2 \Big| Q[t] \Big] \leq C,
\end{equation}
which assumes that the arrival and departure process rates are upper bounded.
According to the fact that $C$ is a constant and the arrival process $a[t]$ is not controllable, the minimization of the upper bound on drift-plus-penalty can be as follows:
\begin{equation}
    V \mathbb{E}\Big[ E(\alpha[t]) \Big] - \mathbb{E}\Big[ Q[t]\cdot b(\alpha[t]) \Big]. \label{eq:dpp_expected}
\end{equation}

Here, the concept of opportunistically minimizing the expectations is used; therefore, \eqref{eq:dpp_expected} can be minimized by an algorithm that observes the current queue state $Q[t]$ and determines $\alpha[t]$ at every slot $t$, as follows:
\begin{equation}
    \alpha^{*}[t]\leftarrow
    \arg\min_{\alpha[t]\in\mathcal{A}}
    \left[
        V\cdot E(\alpha[\tau]) - Q[t]b(\alpha[t])
    \right]
\label{eq:lyapunov-final}
\end{equation}
where $\alpha^{*}[t]$ is the optimal decision at time $t$.

In order to verity whether \eqref{eq:lyapunov-final} works as desired, suppose that $Q[t]=0$. Then, the \eqref{eq:lyapunov-final} tries to minimize $V\cdot E(\alpha[t])$, i.e., the amount of allocated transmit powers should be reduced for energy consumption minimization. This is semantically true because we can focus on the main objective, i.e., energy-efficient computing, because stability is already achieved at this moment. 
On the other hand, suppose that $Q[t]\approx \infty$. Then, the \eqref{eq:lyapunov-final} tries to maximize $b(\alpha[t])$, i.e., the amount of allocated transmit power should be increased for speeding up the service process of $Q[t]$. This is also true because stability should be mainly considered when $Q[t]$ even though we sacrifice certain amounts of energy-efficiency to avoid overflow.

Finally, we confirm that our proposed closed-form mathematical formulation, i.e., \eqref{eq:lyapunov-final}, controls $\alpha[t]$ for minimizing time-average energy consumption subject to queue stability.

\subsubsection{Complexity of the proposed method}\label{sec:algo-complexity}
The proposed method consists of two algorithms. The first algorithm in Sec.~\ref{sec:algo1}, \textit{i.e.}, two-stage mobile charging matching, basically solves two sequential mixed integer ($0$-$1$ binary) convex optimization programs. Since the $0$-$1$ terms (\textit{i.e.}, $x\in\{0, 1\}$) can be relaxed to real values with rounding, our problems can be solvable via pure convex programming, thus resulting in a polynomial-time complexity~\cite{convex,tbc2013joongheon,icc2013joongheon}. 
In addition, the second algorithm in Sec.~\ref{sec:algo2} has the complexity of $O(N)$. 
Finally, we can confirm that the proposed method works in polynomial-time.

\section{Performance Evaluation}\label{sec:4}
This section describes our simulation setup for performance evaluation (refer to Section~\ref{sec:4-1}) and its related evaluation results (refer to Section~\ref{sec:4-2}).

\begin{table}[t!]
\centering
\caption{Simulation parameters~\cite{ref1}}
\begin{tabular}{@{}ll@{}}\toprule
\multicolumn{1}{l}{\centering\textbf{UAV Parameters}} & \multicolumn{1}{l}{\centering\textbf{Value}}\\ \midrule
\multicolumn{1}{l}{Aircraft weight} & \multicolumn{1}{l}{1375 g}\\ 
\multicolumn{1}{l}{Flight speed (max)} & \multicolumn{1}{l}{20 m/s}\\
\multicolumn{1}{l}{Flight time (max)} & \multicolumn{1}{l}{30 min}\\ 
\multicolumn{1}{l}{Capacity of flight battery} & \multicolumn{1}{l}{5870 mAh} \\
\multicolumn{1}{l}{Charging power of flight battery (max)} & \multicolumn{1}{l}{160 W} \\
\multicolumn{1}{l}{Voltage of charger} & \multicolumn{1}{l}{17.4 V} \\
\multicolumn{1}{l}{Rated power of charger} & \multicolumn{1}{l}{100 W} \\
\multicolumn{1}{l}{Charging efficiency $\eta_{k}^{t}$ and $\eta_{j}^{c}$ in (\ref{eq:charging1}) and (\ref{eq:charging2})~\cite{ref2}} & \multicolumn{1}{l}{0.81} \\ \midrule
\multicolumn{1}{l}{\centering\textbf{System Parameters}} & \multicolumn{1}{l}{\centering\textbf{Value}}\\ \midrule
\multicolumn{1}{l}{Map size} & \multicolumn{1}{l}{1299 m $\times$ 750 m}\\
\multicolumn{1}{l}{Number of MBS drones} & \multicolumn{1}{l}{25 (default)}\\
\multicolumn{1}{l}{Number of charging drones} & \multicolumn{1}{l}{50 (default)}\\
\multicolumn{1}{l}{Number of charging towers} & \multicolumn{1}{l}{1 (default)}\\
\multicolumn{1}{l}{Altitude of surveillance UAVs} & \multicolumn{1}{l}{100 m}\\ 
\multicolumn{1}{l}{Number of panels in the charging tower} & \multicolumn{1}{l}{4}\\ 
\multicolumn{1}{l}{Simulation time} & \multicolumn{1}{l}{60 and 120 min}\\ \toprule
\end{tabular}
\label{parameter1}
\end{table}

\subsection{Simulation Setup}\label{sec:4-1}
The performance of the proposed charging and dynamic transmit power allocation method is evaluated via data-intensive simulations. The simulator is designed with \textsf{cvxpy}~\cite{cvxpy}, where parameters used in this paper essentially follow those in~\cite{iotj2018cho,ref1,ref2} and are summarized in Table~\ref{parameter1}.

As performance metrics in the sense of energy saving, we use 1) the {\em residual energy} at each drone that has been served after matching, 2) the {\em coverage-time} (i.e., the time when drones start to drop)~\cite{ton10,vtc06s}, and 3) the {\em queue backlog} that measures the queue stability. Since our charging network model has never been studied in the literature, there is no state-of-the-art method for fair comparison; thus, we adopt a random strategy and two-types of greedy strategies, namely greedy-worst and greedy best strategies, as baselines for our two-stage mobile charging matching. Here, the random strategy performs scheduling at random; and the greedy-worst (greedy-best) strategy allocates more (less) weights to the served drones (i.e., charging drones in the first-stage matching and MBS drones in the second-stage matching) which have higher amounts of residual energy.
Additionally, we adopt the maximum and minimum transmit power allocation strategies, dubbed `Max PA' and  `min PA', respectively, as baseline for our distributed time-average transmit power allocation.

\subsection{Evaluation Results}\label{sec:4-2}

In this subsection, various simulations are carried out to validate the effectiveness of our method in terms of the energy saving and queue stability.

\begin{figure}[t]
    \centering
    \includegraphics[width =0.99\linewidth]{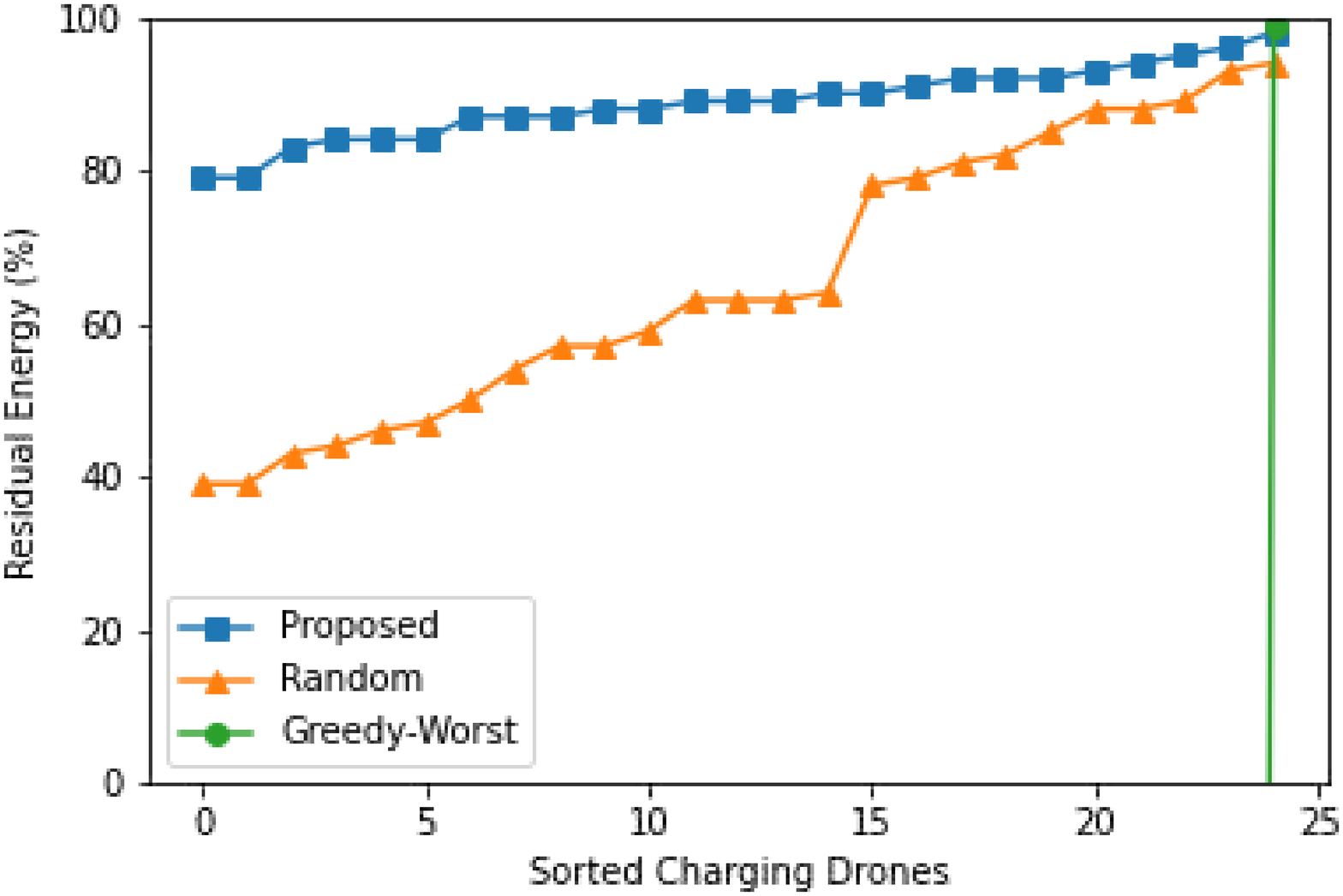}
    \caption{Illustration of residual energy at charging drones in ascending order.}
    \label{fig:sim1}
\end{figure}

\begin{table}[t]
\caption{Statistics of the average amount of residual energy at charging drones}
\label{tab:sim1}
\begin{center}
	\centering
	\begin{tabular}{r||r|r}
    \toprule[1.0pt]
    \centering
      & Proposed & Random \\
    \midrule[1.0pt]
    Average  & $85.08$\% & $55.96$\% \\
    Standard Deviation & $5.48$ & $20.48$ \\
    \bottomrule[1.0pt]
	\end{tabular}
\end{center}
\end{table}

First, we evaluate the performance of the first-stage matching (i.e., matching between charging towers and charging drones).
In Fig.~\ref{fig:sim1}, the residual energy at charging drones is illustrated in ascending order (from the charging drone with the lowest residual energy status up to the one with the highest residual energy status), where 25 charging drones and a single charging tower are deployed in our network.
From Fig.~\ref{fig:sim1}, the following insightful observations are found: 1) it is obvious to see that the proposed method is quite superior to the two baseline strategies in terms of the residual energy; 2) the residual energy tends to be evenly allocated over all charging drones since the optimization in the first-stage matching of our method is formulated as $\left(E_{j}^{c}-e_{j}^{c}[t]\right)$ for charging drone $j$ where $\forall j \in \mathcal{C}$, which takes into account fairness, whereas the random method exhibits a much higher standard deviation of the residual energy among charging drones as it does not account for the fairness; and 3) the greedy-worst strategy consistently charges only one drone by up to the full energy, which implies that the other charging drones do not have an opportunity to be charged, as illustrated in the figure.
The average amount of residual energy at charging drones is summarized in Table~\ref{tab:sim1}. Our method reveals significant energy saving gains over the random strategy while guaranteeing the fairness. More specifically, it is shown from the table that our method offers $52.04$\% and $26.77$\% improvements in terms of the average energy and standard deviation, respectively.

\begin{figure}[t]
    \centering
    \includegraphics[width =0.99\linewidth]{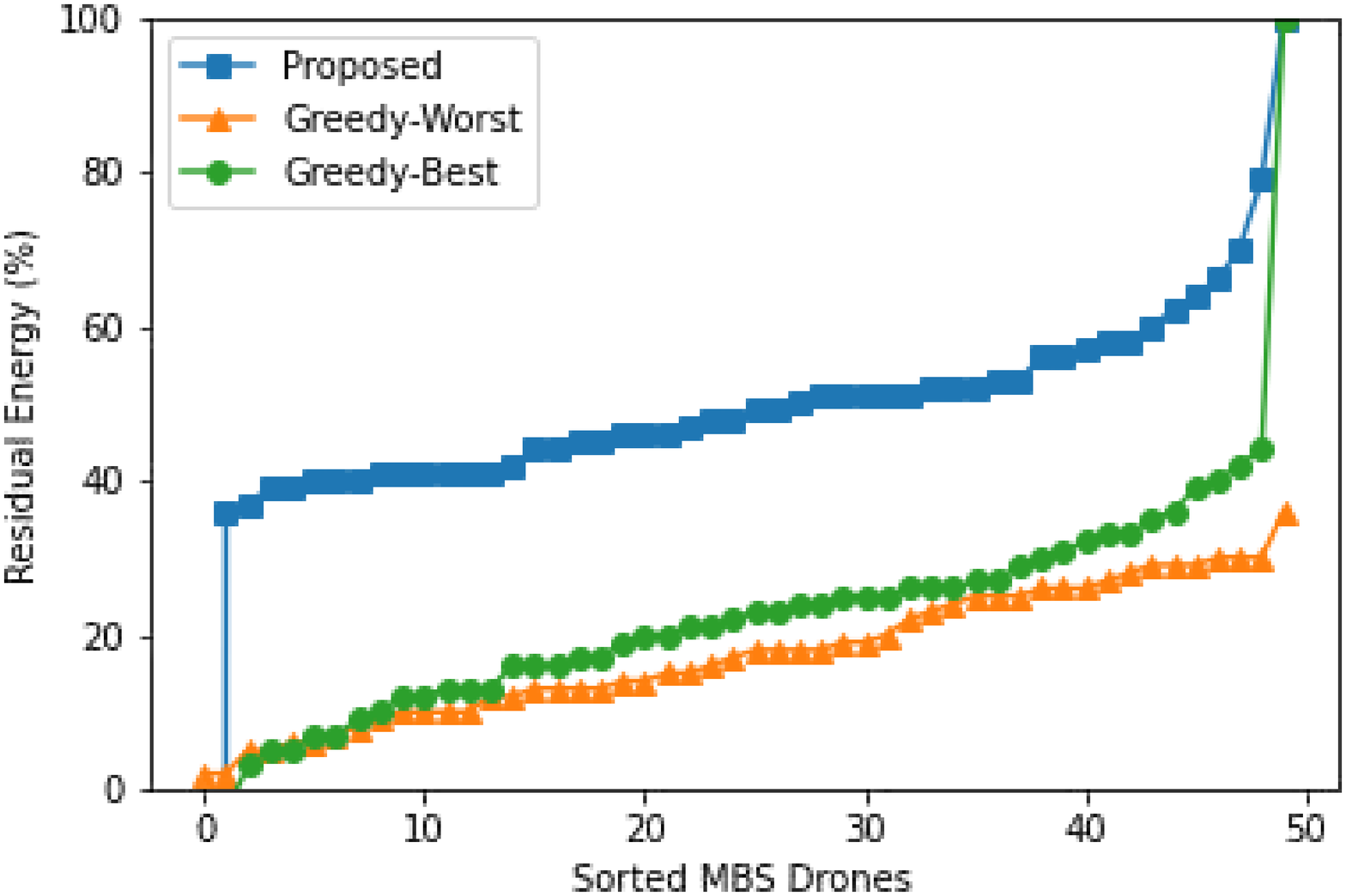}
    \caption{Illustration of residual energy at MBS drones in ascending order.}
    \label{fig:sim2}
\end{figure}

\begin{table}[t]
\caption{Statistics of the average amount of residual energy at MBS drones}
\label{tab:sim2}
\begin{center}
	\centering
	\begin{tabular}{r||r|r|r}
    \toprule[1.0pt]
    \centering
     & Proposed & Greedy-Best & Greedy-Worst \\
    \midrule[1.0pt]
    Average  & $62.36$\% & $26.82$\% & $17.36$\% \\
    \bottomrule[1.0pt]
	\end{tabular}
\end{center}
\end{table}

Second, we turn to evaluating the performance of the second-stage matching (i.e., matching between charging drones and MBS drones).
In Fig.~\ref{fig:sim2}, the residual energy at MBS drones is illustrated in ascending order (from the MBS drone with the lowest residual energy status to the one with the highest residual energy status), where $25$ charging drones and $50$ MBS drones are deployed in our network.
As depicted in the figure, our second-stage matching charges MBS drones much more fairly and efficiently than baseline schemes, where the greedy-best (greedy-worst) method allocates charging drones to MBS drones in ascending (descending) order of the residual energy at each MBS drone. This is because our method is designed in such a way that the value function in~\eqref{eq:value} jointly takes into account (i) the distance between two-types of drones, (ii) the energy status at charging drones $\forall j\in \mathcal{C}$, and (iii) the energy status at MBS drones $\forall i\in \mathcal{M}$.
The average amount of residual energy at MBS drones is summarized in Table~\ref{tab:sim2}. Our second-stage matching also reveals dramatic energy saving gains over the two greedy strategies. More specifically, it is shown from the table that our method offers $132.51$\% and $259.22$\% improvements over the greedy-best and greedy-worst strategies, respectively, in terms of the average energy.

\begin{figure}[t]
    \centering
    \includegraphics[width=0.99\linewidth]{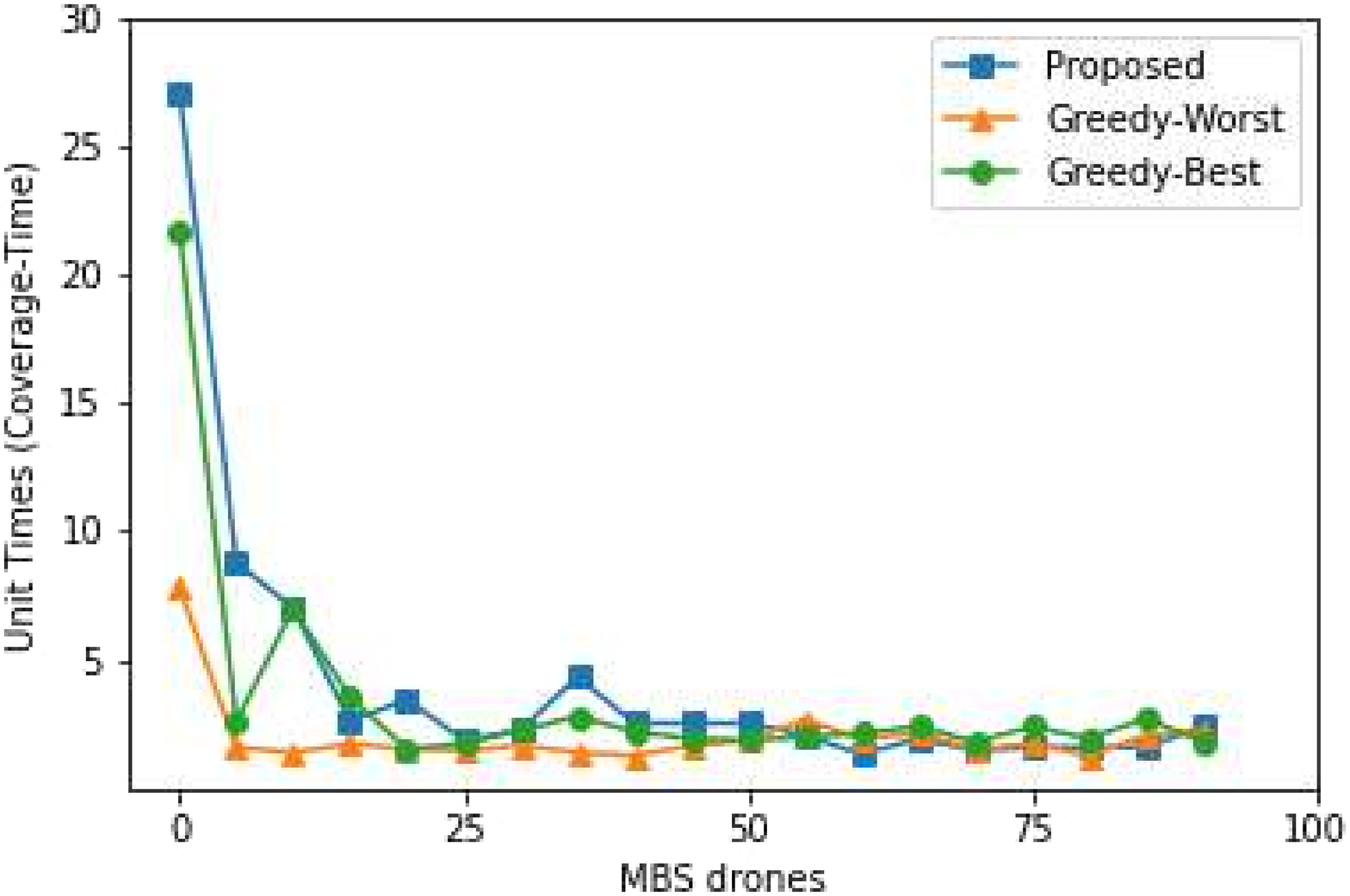}
    \caption{Illustration of the coverage-time according to the number of MBS drones.}
    \label{fig:sim3}
\end{figure}

\begin{figure}[t]
    \centering
    \includegraphics[width=0.99\linewidth]{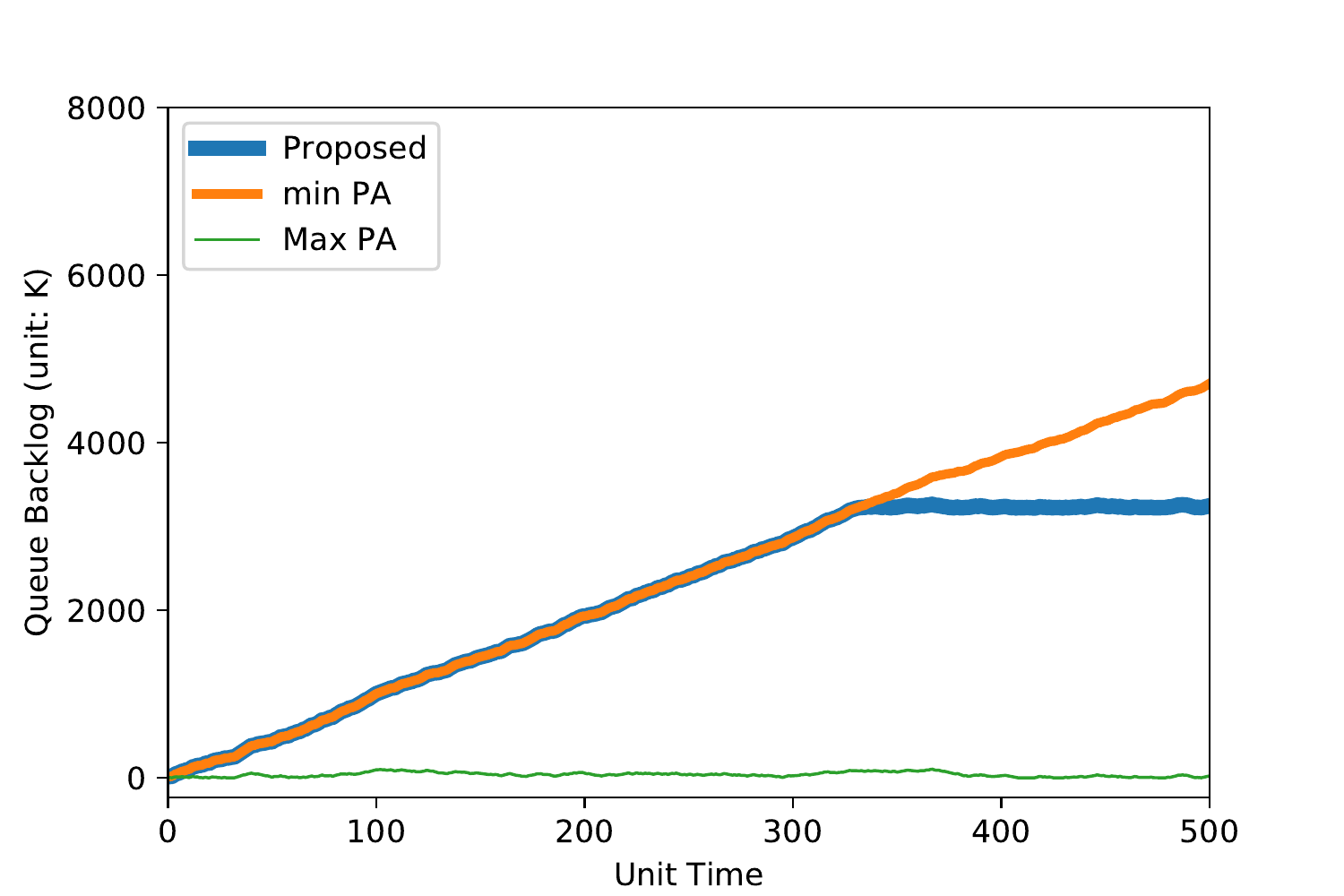}
    \caption{Queue dynamics behaviors in each MBS drone.}
    \label{fig:sim4}
\end{figure}

Third, Fig.~\ref{fig:sim3} illustrates the coverage-time (in the number of unit times)~\cite{ton10,vtc06s} of our method as the number of MBS drones, denoted by $|\mathcal{M}|$, increases, where the number of charging drones, denoted by $|\mathcal{C}|$, is set to $25$ and the number of charging plates in each MBS drone, i.e., $n_{i}^{m}$ in \eqref{eq:s3c2}, is set to $1$. 
In Fig.~\ref{fig:sim3}, $x$-axis and $y$-axis stand for the number of MBS drones (from $1$ up to $50$) and the number of unit times, respectively.
As observed in Fig.~\ref{fig:sim3}, the coverage-time consistently decreases until the number of MBS drones equals to the number of charging drones --- the curve tends to be flattened  after $|\mathcal{M}|=5$ in our setting. That is, when $|\mathcal{M}| >|\mathcal{C}|$, it is highly likely that some MBS drones are not properly scheduled and will drop. The reason behind this is explained as follows. Each MBS drone can be served by charging drones in each unit time; thus, it is possible to maintain a certain amount of energy at each MBS drone. However, due to the distance between MBS and charging drones and the energy status at each charging drone, the status at each MBS drone would be consistently deteriorated over time. It remains open how to further enhance the performance on the coverage-time in such a case (i.e., $|\mathcal{M}| >|\mathcal{C}|$).

Lastly, Fig.~\ref{fig:sim4} illustrates the performance of our energy-aware time-average transmit power allocation in Sec.~\ref{sec:algo2} in comparison with `Max PA' and `min PA' strategies by plotting the queue backlog (unit: $10^{3}$ bits in Fig.~\ref{fig:sim4}) versus the number of unit-times.
From the figure, our insightful findings are described as follows. The queue stability of `Max PA' is consistently guaranteed since the backlog is almost empty in each unit-time at the cost of significant reduction in energy efficiency; thus, the situations of zero delay can be realized. On the other hand, the stability of `min PA' is not achieved as the queue backlog size is monotonically increasing with the number of unit-times, thus resulting in the divergence, which may cause a huge amount of delay, and the built-in queue in each MBS drone may not be long enough due to the physical limitations in storage and processing. That is, queue overflow would occur beyond a certain unit-time when `min PA' is employed.
On the contrary, our dynamic time-average transmit power allocation method is capable of guaranteeing both queue stability refer to Fig.~\ref{fig:sim4} along with bounded delays and energy efficient successfully. More specifically, in the early stages, our method tends to allocate a small mount of power for energy efficiency since the queue backlog is not long yet. As the number of unit-times increases, the queue backlog reaches at a certain point where the transmit power needs to be controlled in order to stabilize the queue status --- it is observed from Fig.~\ref{fig:sim4} that the queue backlog becomes stable after 340 unit-times. In consequence, our method operates balancing appropriately between the queue stability and the energy efficiency, which thus enables our network to extend the coverage-time as well.

\section{Conclusions and Future Work}\label{sec:5}
This paper proposes a novel joint mobile charging and coverage-time extension method for drone-enabled future cellular networks. 
First of all, a two-stage mobile charging matching algorithm is designed and optimized where the first stage is for matching between charging towers and charging drones and the second stage is for matching between charging drones and MBS drones. In this case, the second stage needs to optimize the matching/scheduling and the allocation of powers in each scheduled pair, i.e., non-convex terms exist, as discussed. In this paper, we convert the given non-convex terms into convex terms, and then eventually, we formulate the given problem as convex programming which can guarantee optimal solutions. Next, we also design distributed time-average optimal transmit power allocation algorithm subject to queue stability in each MBS drones, inspired by Lyapunov optimization theory -- drift-plus-penalty algorithm. 
As presented in performance evaluation via data-intensive simulations, our proposed method achieves desired performance improvement, in terms of i) energy efficiency in each charging drone, ii) energy efficiency in each MBS drone, and iii) coverage-time extension in each MBS drone.

We remark that the proposed algorithm in this paper has built purely upon the energy-efficient operations for coverage-time extension in UAV networks. As a future avenue, on top of this proposed energy-aware framework, it is worthy to consider additional performance enhancement in terms of sum rate maximization, quality of services (QoS), quality of experience (QoE), and so forth.



\begin{IEEEbiography}[{\includegraphics[width=1in,height=1.25in,clip,keepaspectratio]{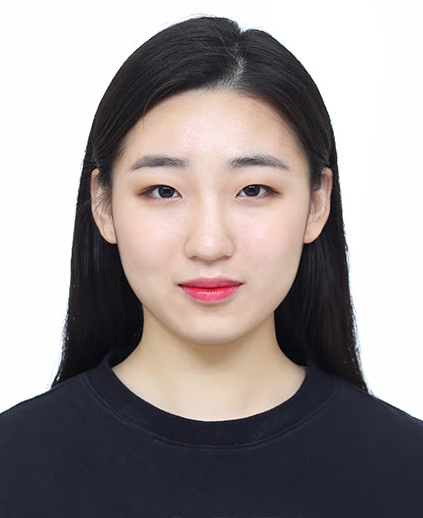}}]{Soohyun Park}
is currently pursuing her Ph.D. degree in electrical engineering at Korea University, Seoul, Republic of Korea. 
Her research focuses include deep learning algorithms and their applications. 
She was a recipient of the IEEE Vehicular Technology Society (VTS) Seoul Chapter Award in 2019.
\end{IEEEbiography}

\begin{IEEEbiography}[{\includegraphics[width=1in,height=1.25in,clip,keepaspectratio]{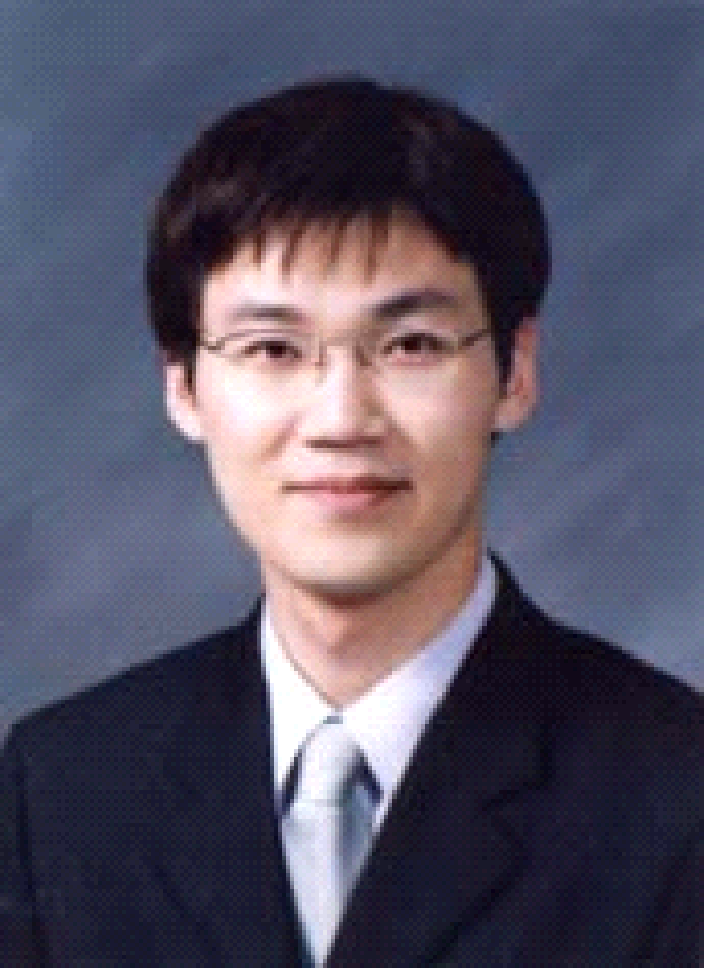}}]{Won-Yong Shin} (S’02-M’08-SM’16) received the B.S. degree in electrical engineering from Yonsei University, Seoul, Republic of Korea, in 2002. He received the M.S. and the Ph.D. degrees in electrical engineering and computer science from the Korea Advanced Institute of Science and Technology (KAIST), Daejeon, Republic of Korea, in 2004 and 2008, respectively. In May 2009, Dr. Shin joined the School of Engineering and Applied Sciences, Harvard University, Cambridge, MA USA, as a Postdoctoral Fellow and was promoted to a Research Associate in October 2011. From March 2012 to February 2019, he was a Faculty Member of the
Department of Computer Science and Engineering, Dankook University, Yongin, Republic of Korea. Since March 2019, he has been with the School of Mathematics and Computing (Computational Science and Engineering), Yonsei University, Seoul, Republic of Korea, where he is currently an Associate Professor. His research interests are in the areas of information theory, communications, signal processing, mobile computing, big data analytics, and online social networks analysis.

From 2014 to 2018, Dr. Shin served as an Associate Editor of the
\textit{IEICE Transactions on Fundamentals of Electronics, Communications
and Computer Sciences}. He also served as an Organizing Committee
Member for the 2015 IEEE Information Theory Workshop. He received
the Bronze Prize of the Samsung Humantech Paper Contest (2008) and
the KICS Haedong Young Scholar Award (2016).
\end{IEEEbiography}

\begin{IEEEbiography}[{\includegraphics[width=1in,height=1.25in,clip,keepaspectratio]{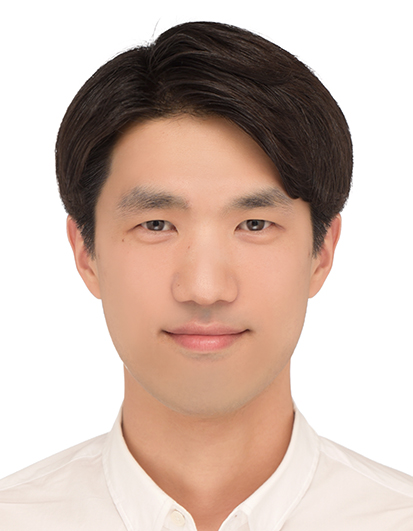}}]{Minseok Choi} has been an assistant professor with Jeju National University, Jeju, Korea, since 2020.
He received the B.S., M.S. and Ph.D. degrees in the School of Electrical Engineering from Korea Advanced Institute of Science and Technology (KAIST), Daejeon, Korea, in 2011, 2013, and 2018, respectively. 
He was a visiting postdoctoral researcher in electrical and computer engineering at the University of Southern California (USC), Los Angeles, CA, USA, and a research professor in electrical engineering at the Korea University, Seoul, Korea.
His research interests include wireless caching networks, stochastic network optimization, non-orthogonal multiple access, and 5G networks.
\end{IEEEbiography}

\begin{IEEEbiography}[{\includegraphics[width=1in,height=1.25in,clip,keepaspectratio]{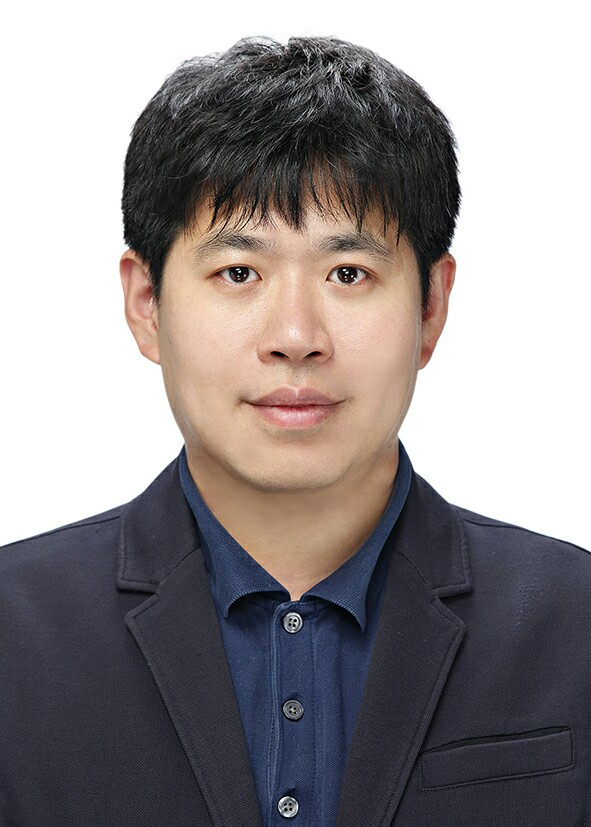}}]{Joongheon Kim} (M'06--SM'18) has been with Korea University, Seoul, Korea, since 2019, and he is currently an associate professor at the Department of Electrical and Computer Engineering. He is also a vice director of Artificial Intelligence Engineering Research Center and a dean of Center for Teaching and Learning at Korea University, Seoul, Korea. 
He received the B.S. and M.S. degrees in Computer Science and Engineering from Korea University, Seoul, Korea, in 2004 and 2006, respectively; and the Ph.D. degree in Computer Science from the University of Southern California (USC), Los Angeles, CA, USA, in 2014. 
Before joining Korea University, he was with LG Electronics (Seoul, Korea, 2006--2009), InterDigital (San Diego, CA, USA, 2012), Intel Corporation (Santa Clara in Silicon Valley, CA, USA, 2013--2016), and Chung-Ang University (Seoul, Korea, 2016--2019). 

He is a senior member of the IEEE, and serves as an associate editor for \textsc{IEEE Transactions on Vehicular Technology}. He internationally published more than 90 journals, 110 conference papers, and 6 book chapters. He also holds more than 50 patents, majorly for 60\,GHz millimeter-wave IEEE 802.11ad and IEEE 802.11ay standardization. 

He was a recipient of Annenberg Graduate Fellowship with his Ph.D. admission from USC (2009), 
Intel Corporation Next Generation and Standards (NGS) Division Recognition Award (2015), Haedong Young Scholar Award by KICS (The Korean Institute of Communications and Information Sciences) (2018), IEEE Vehicular Technology Society (VTS) Seoul Chapter Award (2019), Outstanding Contribution Award by KICS (2019), Paper Awards from IEEE Seoul Section Student Paper Contests (2019, 2020), Best Teaching Awards by Korea University (Top 5\% in Fall-2019, Top 20\% in Fall-2020), \textsc{IEEE Systems Journal} Best Paper Award (2020), IEEE ICOIN Best Paper Award (2021), and Haedong Paper Award by KICS (2021).
\end{IEEEbiography}

\EOD

\end{document}